\documentclass[british]{scrartcl}
\usepackage[T1]{fontenc}
\usepackage[latin9]{inputenc}
\usepackage{float}
\usepackage{bm}
\usepackage{amsmath}
\usepackage{amssymb}
\usepackage{graphicx}
\usepackage{esint}
\usepackage[authoryear]{natbib}
\usepackage{amsthm}
\usepackage{url}
\usepackage{color}
\usepackage{enumitem}
\usepackage{stmaryrd}

\makeatletter

\floatstyle{ruled}
\newfloat{algorithm}{tbp}{loa}
\providecommand{\algorithmname}{Algorithm}
\floatname{algorithm}{\protect\algorithmname}

\newcounter{rmq}[section]
\newcommand{\R}{\mathbb{R}}

\newcommand{\naturals}{\mathbb{N}}

\newcommand{\setX}{\mathbb{X}}
\newcommand{\sigX}{\mathcal{X}}
\newcommand{\setY}{\mathbb{Y}}
\newcommand{\sigY}{\mathcal{Y}}

\renewcommand{\P}{\mathbb{P}}

\newcommand{\E}{\mathbb{E}}

\newcommand{\ass}{\mbox{$\mathbb{P}$-a.s.}}

\newcommand{\F}{\mathcal{F}} 

\newcommand{\ind}{\mathbf{1}}
\newcommand{\dd}{\mathrm{d}}

\newcommand{\opnm}{\interleave}

\newcommand{\comment}[1]{ \ifthenelse{ \equal{\showcomment}{true} }{ {\bf #1} }{} }
\newcommand{\showcomment}{true}
\newcommand{\subscript}[2]{$#1 _ #2$}

\newtheorem{thm}{Theorem}
\newtheorem{prop}{Proposition}

\newtheorem{lemma}{Lemma}
\newtheorem{defn}{Definition}
\newtheorem{remark}{Remark}

\makeatother

\usepackage{babel}

\title{Stability with respect to initial conditions in V-norm for nonlinear filters with ergodic observations}

\begin{document}

\author{Mathieu Gerber
\and Nick Whiteley}

\date{\textit{Harvard University   and University of Bristol}}

\maketitle

\begin{abstract}


We establish conditions for an exponential rate of forgetting of the initial distribution of nonlinear filters in $V$-norm, path-wise along almost all observation sequences. In contrast to previous works, our results allow for unbounded test functions. The analysis is conducted in an general setup involving nonnegative kernels in a random environment which allows treatment of filters and prediction filters in a single framework. The main result is illustrated on two examples, the first showing that a total variation norm stability result obtained by \citet{Douc2009} can be extended to $V$-norm without any additional assumptions, the second concerning a situation in which  forgetting of the initial condition holds in $V$-norm for the filters, but the $V$-norm of each prediction filter is infinite.

\textit{Keywords:} Nonlinear filtering, Hidden Markov models, random environment, $V$-norm

\end{abstract}

\section{Introduction}\label{sec:intro}

For Polish spaces $\setX$, $\setY$ equipped with their Borel $\sigma$-algebras $\sigX$, $\sigY$, let $\mu$ be a probability measure on $\sigX$ and let $f:\setX\times\sigX\rightarrow[0,1]$ and $g:\setX\times\sigY\rightarrow[0,1]$ be probability kernels. A hidden Markov model (HMM) is a bi-variate process $(X,Y)$ where the signal process $X=(X_n)_{n\in\naturals}$ is a Markov chain with initial distribution $\mu$ and transition kernel $f$, and the observations $Y=(Y_n)_{n\in\naturals}$ are conditionally independent given $X$, with the conditional distribution of $Y_n$ given $X$ being $g(X_n,\cdot)$. The filtering problem is to compute, for each $n$, the conditional distribution of $X_n$ given $Y_0,\ldots,Y_n$.

Suppose that for each $x\in\setX$, $g(x,\cdot)$ admits a density denoted $g(x,y)$  w.r.t. some $\sigma$-finite measure. Then for a probability measure $\lambda$ on $\sigX$, under mild conditions on the density $g$ the following recursion defines a sequence of probability kernels $\Pi_n^{\lambda}:\setY^{\naturals}\times\sigX\to[0,1]$, $n\geq0$,
\begin{align}
\label{eq:filter_defn}
\begin{split}
\Pi_n^{\lambda}(y,A) & :=\frac{\int\ind_A(x)g(x,y_n)f(x^\prime,\dd x)\Pi_{n-1}^{\lambda}(y,\dd x^\prime)}{\int g(x,y_0)f(x^\prime,\dd x)\Pi_{n-1}^{\lambda}(y,\dd x^\prime)},\quad n\geq 1,\\
\Pi_0^{\lambda}(y,A) & :=\frac{\int\ind_A(x)g(x,y_0)\lambda(\dd x)}{\int g(x,y_0)\lambda(dx)},
\end{split}
\end{align}
where $y=(y_0,y_1,\ldots)\in \setY^{\naturals}$. In particular, $\Pi_n^{\mu}(Y,\cdot)$ is a version of the conditional distribution of $X_n$ given $Y_0,\ldots,Y_n$ under the probability model described in the first paragraph of this section. A distribution of the form $\Pi_n^{\lambda}(y,\cdot)$ is called a \emph{filtering distribution}, or simply a \emph{filter}.

The question of stability w.r.t. initial conditions of the filter addresses whether or not $\Pi_n^{\lambda}$ is, in some sense, insensitive to $\lambda$ as $n\to+\infty$. This question has been made precise in a number of ways and answered using a variety of techniques. A full survey of existing results is beyond the scope of this article, the reader is directed to \citet[Chapter 4]{crisan2011oxford} for a collection of recent perspectives. However, what unifies much of the literature on filter stability is that the insensitivity of $\Pi_n^{\lambda}$ to $\lambda$ is described in terms of integrals  w.r.t.  $\Pi_n^{\lambda}(y,\cdot)$ of \emph{bounded} test functions, typically through decay as $n\to+\infty$ of the total variation $\|\Pi_n^{\lambda}(Y,\cdot)-\Pi_n^{\tilde{\lambda}}(y,\cdot)\|_{\text{tv}}$,  where for a signed measure $m$, $\|m\|_{\text{tv}}:=\sup_{|\varphi|\leq 1}|m(\varphi)|$, $m(\varphi):=\int\varphi(x)m(\dd x)$. Studies using the total variation norm include e.g., \citet{kleptsyna2008discrete,Douc2009,douc2010forgetting, van2009stability} and several older and influential works deal with bounded continuous test functions e.g., \citet{ocone1996asymptotic} for filtering in continuous time.

In many applications $\setX$ is $\R^d$ or some other unbounded domain and the motive for computing $\Pi_n^{\lambda}$ is statistical inference for the signal process, e.g. by calculating moments of $\Pi_n^{\lambda}(y,\cdot)$. This situation leads naturally to the question of filter stability for unbounded test functions, which to the knowledge of the authors has gone largely unanswered. The main aim of the present article is to address this gap.

One situation in which unbounded test functions are implicitly considered is the linear-Gaussian case, in which $\setX$, $\setY$ are Cartesian products of $\R$, $x$ is a vector and $f(x,\cdot)=\mathcal{N}(Ax,C)$, $g(x,\cdot)=\mathcal{N}(Bx,D)$ for suitable matrices $A,B,C,D$. If $\lambda$ is also Gaussian, then $\Pi_n^{\lambda}(y,\cdot)$ is too. In this situation the stability of the mean and covariance of $\Pi_n^{\lambda}(y,\cdot)$ w.r.t. initial conditions has been well studied, see  \citet{ocone1996asymptotic} for continuous time and references therein for discrete time, but the proof techniques involved seem very much specific to the linear-Gaussian model structure.

Our aim is to accommodate unbounded test functions whilst allowing more general assumptions on $\lambda$, $f$ and $g$. Our approach builds very directly upon that of \cite{Douc2009}, in turn drawing on techniques of \citet{kleptsyna2008discrete}. Under a collection of assumptions which we discuss in more detail later, \citet[Theorem 1]{Douc2009} established path-wise exponential stability of the form: there exists a strictly positive constant $c$ such that for any two probability measures $\lambda,\tilde{\lambda}$:
\begin{equation}
\limsup_{n\to+\infty} \frac{1}{n}\log \|\Pi_n^{\lambda}(Y,\cdot)-\Pi_n^{\tilde{\lambda}}(Y,\cdot)\|_{\text{tv}} <-c,\quad \mathbb{P}-a.s.,\label{eq:DM_result_front}
\end{equation}
where $\mathbb{P}$ is a probability measure on $\sigY^{\otimes\naturals}$. The conditions of \citet[Theorem 1]{Douc2009} accommodate  $Y$ being a stationary ergodic process under $\mathbb{P}$, and allow for model mis-specification, in the sense that $\mathbb{P}$  need not be the measure on $\sigY^{\otimes\naturals}$ induced by the HMM specified by $\mu,f,g$, or indeed any HMM.

Our main contribution is to establish that under conditions similar to those of \cite{Douc2009}, path-wise exponential convergence as in \eqref{eq:DM_result_front} holds, but with $\|\cdot\|_{\text{tv}}$ replaced by a norm which allows for unbounded test functions. For $V$ an $\R^+$-valued function on $\setX$ such that $\sup_{x\in\setX}V(x)\leq +\infty$, we consider the norm on signed measures $m$, $\|m\|_V:=\sup_{|\varphi |\leq V} | m(\varphi)|$,
as is popular in studies of the ergodic properties of general state-space Markov chains \citep[Chapter 16]{meyntweedie}. Stability w.r.t. initial conditions of the prediction filters  $\overline{\Pi}_{n}^\lambda(y,\cdot):=\int f(x,\cdot)\Pi_{n-1}^\lambda(y,\dd x)$ in $V$-norm was considered by the second author of the present paper in \citep{Whiteley2013}, but under a restrictive condition on the observation sequence $y$ which we discuss in more detail later, after introducing notation and definitions.

The rest of the paper is structured as follows. Our general setup is introduced in Section \ref{sub:defn_and_assumptions}, where we take a slightly abstract perspective on HMM's and filtering, in that  our main results concern certain sequences of measures which arise from the composition of nonnegative kernels driven by an ergodic measure-preserving transform, along the lines considered by \cite{kifer1996perron} in his Perron-Frobenius theorem in random environments. The utility of this general formulation is that it allows us to treat stability of the filters $\Pi_n^\lambda$ and the prediction filters $\overline{\Pi}_n^\lambda$ in a single framework, as described in Section \ref{sub:instances}. The statements of our main results are given in Section \ref{sub:statements}, with comparisons to the assumptions of \cite{Douc2009} and a result of \cite{Whiteley2013}. Verification of the assumptions is illustrated through examples in Section \ref{sub:examples} and we point to potential other applications in Section \ref{sub:concluding_remarks}. The proofs are in Section \ref{sec:proofs}.

\section{Nonnegative kernels in a random environment}\label{sec:generic}

\subsection{Definitions and assumptions}\label{sub:defn_and_assumptions}

We consider a complete probability space $(\Omega,\mathcal{F},\mathbb{P})$   and a measurable space $(\setX,\sigX)$, where $\setX$ is Polish and $\sigX$ is the Borel $\sigma$-algebra on $\setX$. For an integral kernel $R:\Omega\times\setX\times\sigX\to [0,+\infty]$, i.e. for  $(\omega,x)\in\Omega\times\setX$, $R(\omega,x,\cdot)$ is a measure on $\sigX$, and for $A\in\sigX$, $R(\cdot,\cdot,A)$ is measurable w.r.t. $\mathcal{F}\otimes\sigX$,  we shall write interchangeably $R(\omega,x,A)\equiv R^\omega(x,A)$. Similarly for $\nu:\Omega\times\sigX\to[0,+\infty]$,  $\nu(\omega,A)\equiv\nu^\omega(A)$; for $\varphi:\Omega\times\setX\to\mathbb{R}$, $\varphi(\omega,x)\equiv\varphi^\omega(x)$; and $R^\omega\varphi^\omega(x):=\int_{\setX}\varphi^\omega(x^\prime)R^\omega(x,\dd x^\prime)$, $\nu^\omega R^\omega(\cdot):=\int_{\setX}\nu^\omega(\dd x)R^\omega(x,\cdot)$, $\nu^\omega(\varphi^\omega):=\int_{\mathbb{X}}\varphi^\omega(x)\nu^\omega(\dd x)$.

By virtue of our completeness assumption about $(\Omega,\mathcal{F},\mathbb{P})$ and Polish assumption about $\setX$, for any measurable $\varphi:\Omega\times\setX\to\mathbb{R}$ and $A\in\sigX$, the mappings $\omega\mapsto \sup_{x\in A} \varphi(\omega,x)$ and $\omega\mapsto \inf_{x\in A} \varphi(\omega,x)$ are each measurable w.r.t. $\mathcal{F}$ \cite[Corollary 2.13]{Crauel2003}.

We fix a function $V:\setX\to[1,+\infty)$ possibly unbounded (in the sense that we allow $\sup_{x\in\setX }\leq +\infty$), with which we associate the following norms. For $\varphi:\setX\to\R$, $\|\varphi\|_V:=\sup_{x\in\setX}  |\varphi(x)|/V(x)$; for any signed measure $m$ on $\sigX$, $\|m\|_V:=\sup_{\varphi:|\varphi|\leq V}|m(\varphi)|$; and for any two nonnegative integral kernels $R,\widetilde{R}$ on $(\mathbb{X},\mathcal{X})$, $\opnm R - \widetilde{R} \opnm_V:= \sup_{x\in\setX} \|R(x,\cdot)-\widetilde{R}(x,\cdot)\|_V/V(x) $.

Let $\theta:\Omega\rightarrow\Omega$ be a measurable mapping and with $n\in\naturals$, let $\theta^n$ denote the $n$-fold iterate of $\theta$. Then denote:
\[
R_{0}^{\omega}:=Id,\qquad R_{n}^{\omega}:=R^{\omega}R^{\theta\omega}\cdots R^{\theta^{n-1}\omega},\quad n\geq1.
\]
Define $a \vee b :  = \max\{a,b\}$, $a \wedge b := \min\{a,b\}$, $\log^+(x):=\log(1\vee x)$ and $\log^-(x) := - \log (1\wedge x) $. The indicator function on a set $A$ is denoted by $\ind_A$. The set of non-negative integers is denoted $\mathbb{N}$. We adopt the conventions $0/0=+\infty/+\infty=1$.

From henceforth we fix a distinguished nonnegative kernel $Q:\Omega\times\setX\times\sigX\to[0,+\infty]$, such that $Q^\omega(x,\setX)>0$ for all $x\in\setX$, $\P$-a.s.

\begin{defn}
A set $C\in\mathcal{X}$ is a local Doeblin (LD) set for $Q$ if there exist
nonnegative random variables $\epsilon_{C}^{-},\epsilon_{C}^{+}$ on $(\Omega,\mathcal{F})$,  such that $\epsilon_{C}^{-}(\omega)\leq \epsilon_{C}^{+}(\omega)$ for all $\omega$ and both $\epsilon_{C}^{+}$ and $\epsilon_{C}^{-}$ are valued in $(0,+\infty)$ $\P$-a.s.; and
a probability kernel $\mu_{C}:\Omega\times\sigX\to[0,1]$ such that $\mu_{C}^{\omega}(C)=1$ for
all $\omega$ and, for any $(A,x,\omega)\in\sigX\times C\times \Omega$,
\[
\epsilon_{C}^{-}(\omega)\mu_{C}^{\omega}(A\cap C)\leq Q^{\omega}(x,A\cap C)\leq\epsilon_{C}^{+}(\omega)\mu_{C}^{\omega}(A\cap C).
\]

\end{defn}

We shall consider the following assumptions.
\begin{enumerate}[label=(\subscript{A}{\arabic*})]
\item\label{H:theta} $\theta$ preserves $\mathbb{P}$ and is ergodic;

\item\label{H:Upsilon} $\mathbb{E}\left[\log^+ \Upsilon\right]<+\infty$, where $\Upsilon(\omega):=\opnm Q^\omega \opnm_V$, $\omega\in\Omega$;

\item \label{H:Psi}  There exists a set $D\in\sigX$ such that $\mathbb{E}\left[\log^- \Psi\right]<+\infty$, where $\Psi(\omega):=\inf_{x\in D}Q^{\omega}(x,D)$, $\omega\in\Omega$; 


\item\label{H:drift} There exist a set $K\in\mathcal{F}$, a constant $\underline{d}\geq 0$,  and a measurable, unbounded function $W:\setX\rightarrow[0,+\infty)$
such that for any $d\in[\underline{d},+\infty)$,

\begin{enumerate}
\item  $C_{d}:=\{x\in\setX\,:\, W(x)\leq d\}$
is a LD set for $Q$, with
$$
\inf_{\omega \in K}\epsilon_{C_d}^-(\omega)/\epsilon_{C_d}^+(\omega)\in (0,1],\quad \E\big[\log^-\big(\epsilon_{C_d}^-\mu_{C_d}(C_d\cap D)\big)\big]<+\infty,
$$
where $D$ is as in \ref{H:Psi};

\item$V_d:=\sup_{x\in C_{d}}V(x)<+\infty$ and
\[
\frac{Q^{\omega}(V)(x)}{V(x)}\leq\exp\left[ -W(x)\right],\quad\forall(\omega,x)\in K\times C_d^c;
\]
\end{enumerate}

\item \label{H:setK} $\P(K)>2/3$, where $K$ is as in \ref{H:drift}.


\end{enumerate}

Let $\mathcal{M}(D,V)$ be the collection of integral kernels $\nu:\Omega\times\sigX\to[0,+\infty]$, such that for any $A\in\sigX$, the mapping $\omega\mapsto\nu^\omega(A)$ is measurable; for $\P$-almost all $\omega$, $\nu^\omega(\cdot)$ is a probability measure on $\sigX$,  $\nu^\omega(V)<+\infty$ and $\nu^\omega Q^\omega(D)>0$,  where $D$ is as in \ref{H:Psi}.

For any integral kernel $\nu:\Omega\times\sigX\to[0,+\infty]$ and $n\in\naturals$, denote:
\begin{equation}
\eta_{\nu,n}^\omega(A):=\frac{\nu^\omega Q_n^\omega (A)}{\nu^\omega Q_n^\omega(\mathbb{X})},\quad(\omega,A)\in \Omega\times\mathcal{X}.\label{eq:eta_defn}
\end{equation}

The following preliminary lemma addresses some basic regularity properties of $Q$ and $\eta_{\nu,n}$.
\begin{lemma}\label{lemma:eta_prelim}
Assume \ref{H:theta}, \ref{H:Upsilon}, \ref{H:Psi} and let $\nu\in\mathcal{M}(D,V)$. Then there exists $\bar{\Omega}\in\mathcal{F}$ with $\P(\bar{\Omega})=1$ such that the following hold for all $\omega\in\bar{\Omega}$. For all $n\in\naturals$, $\opnm Q_n^\omega \opnm_V<+\infty$, $\prod_{k=0}^n\Upsilon(\theta^k\omega)<+\infty$, $\prod_{k=0}^n\Psi(\theta^k\omega)>0$, and for all $x\in\setX$, $Q_n^{\omega}(x,\setX)>0$. Also $\eta_{\nu,n}^\omega(\cdot)$ is a probability measure on $\sigX$ and $\eta_{\nu,n}^\omega(V)<+\infty$ for all $n\in\naturals$, and $\eta_{\nu,n}^\omega(D)>0$ for all $n\geq1$.
\end{lemma}
\begin{proof}
By the sub-multiplicative property of $\opnm \cdot \opnm_V$, we have $\opnm Q_n^\omega \opnm_V\leq \prod_{k=0}^{n-1} \opnm Q^{\theta^k\omega} \opnm_V = \prod_{k=0}^{n-1}\Upsilon(\theta^k\omega)$.  For $\P$-almost all $\omega$, $\Upsilon(\omega)<+\infty$ by \ref{H:Upsilon},  $\Psi(\omega)>0$ by \ref{H:Psi}, and $Q^\omega(x,\setX)>0$ for all $x\in\setX$  by definition. Combining these observations with the measure preservation part of \ref{H:theta} gives the first four inequalities in the statement. For $n\in\naturals$, $\|\nu^\omega Q_n^\omega\|_V\leq\|\nu^\omega \|_V\opnm Q_n^\omega \opnm_V<+\infty$ and for $n\geq1$, $\nu^\omega Q_n^\omega (\mathbb{X})\geq \nu^\omega Q^\omega(D)\prod_{k=0}^{n-1}1\wedge\Psi(\theta^k\omega)>0$, $\P$-a.s. Putting these facts together with \eqref{eq:eta_defn} completes the proof.
\end{proof}

\subsection{Instances of the general setup}\label{sub:instances}

Let $(\setY,\sigY)$, $f$, $g$, $\lambda$, $\Pi_n^\lambda$ and $\overline{\Pi}_n^\lambda$ be as in Section \ref{sec:intro} and let $\P$ be some probability measure on $\sigY^{\otimes\naturals}$. Take $\Omega=\setY^\naturals$. We note that the requirement of Section \ref{sub:defn_and_assumptions} to have a complete probability space can always be satisfied by taking $(\Omega,\mathcal{F},\P)$ to be the unique completion of $(\Omega,\sigY^{\otimes\naturals},\P)$ in the sense of \citet[p.39 and problem 3.5, p.43]{Bill2ndEd}, where here and in the examples of Section \ref{sub:examples} we abuse notation slightly in using the symbol $\P$ to represent both the original probability measure on $\sigY^{\otimes\naturals}$ and its extension to $\mathcal{F}$.

Regard $Y(\omega)=(Y_0(\omega),Y_1(\omega),\ldots)$ as the coordinate process on $(\Omega,\mathcal{F})$. Take $\theta$ as the shift operator, $Y(\theta\omega) = (Y_1(\omega),Y_2(\omega),\ldots)$. We then observe the following from \eqref{eq:filter_defn} and \eqref{eq:eta_defn}.

\begin{description}
  \item[Filters] If  one takes
\begin{equation}
\nu^\omega(\dd x)=\frac{g(x,Y_0(\omega))\lambda(\dd x)}{\int g(x^\prime,Y_0(\omega))\lambda(\dd x^\prime)},\quad Q^\omega(x,\dd x^\prime) = f(x,\dd x^\prime) g(x^\prime,Y_1(\omega)),\label{eq:filters}
\end{equation}
then  $\eta_{\nu,n}^\omega(\cdot) \equiv \Pi_n^\lambda(Y(\omega),\cdot)$.
  \item[Prediction filters] If one takes
\begin{equation}
\nu^\omega(\cdot)=\lambda(\cdot),\quad Q^\omega(x,\dd x^\prime) = g(x,Y_0(\omega))f(x,\dd x^\prime),\label{eq:pred_filters}
\end{equation}
then $\eta_{\nu,n}^\omega(\cdot) \equiv \overline{\Pi}_n^\lambda(Y(\omega),\cdot)$.
\end{description}

\subsection{Statements of the main results}\label{sub:statements}

\begin{thm}\label{thm:conv_to_zero}
Assume \ref{H:theta}-\ref{H:setK}. Then there exists a $\rho\in(0,1)$ such that, for all $\nu,\,\tilde{\nu}\in\mathcal{M}(D,V)$,
$$
\lim_{n\rightarrow+\infty} \rho^{-n}\|\eta_{\nu,n}^{\omega}-\eta_{\tilde{\nu},n}^{\omega}\|_V=0,\quad\ass
$$
\end{thm}
The main ingredients in the proof of Theorem \ref{thm:conv_to_zero} are the following two propositions.  Proof of Proposition \ref{prop:Vnorm} is the main technical contribution of the paper and is given  in Section \ref{sec:proof_of_vnorm} through a sequence of lemmas.  The proof of Proposition \ref{prop:forget},  given in Section \ref{sec:proof_of_forget},  follows quite closely some arguments of \citet[][Proof of Proposition 5]{Douc2012}, with suitable modifications to accommodate the $V$-norm. Lastly, we note although we are primarily interested in the case $\sup_x V(x)=+\infty$, none of our results actually require that condition to hold, and when e.g., $V(x)=1$ the $V$-norm on measures reduces to the tv-norm and the claim of Proposition \ref{prop:Vnorm} is trivial.

\begin{prop}\label{prop:Vnorm}
Assume \ref{H:theta}-\ref{H:drift}, $\P(K)>0$, with $K$ as in \ref{H:drift}. Then,   for any $\beta\in (0,1)$ and $\nu\in\mathcal{M}(D,V)$,
$$
\lim_{n\rightarrow+\infty}\beta^n \|\eta_{\nu,n}^{\omega}\|_{V}=0,\quad\ass
$$
\end{prop}

\begin{prop}\label{prop:forget}
Assume \ref{H:theta}-\ref{H:drift}, let $\nu,\tilde{\nu}$ be two members of $\mathcal{M}(D,V)$, and let $\gamma^-,\gamma^+,\beta$ be constants such that $0\leq\gamma^-<\gamma^+\leq 1$ and $\beta\in (\gamma^-,\gamma^+)$. Fix any $d\in [\underline{d},+\infty)$. Then for $\P$-almost any $\omega$, if $n^{-1} I_{0,n-1}^{\omega}\geq  ( 1-\gamma^-)\vee (1+\gamma^+)/2$,  then:
\begin{multline}\label{eq:lem1}
\|\eta^{\omega}_{\nu,n}-\eta^{\omega}_{\tilde{\nu},n}\|_V
\\
\leq 2\rho_{C_d}^{\lfloor n (\beta-\gamma^-)\rfloor}\|\eta^{\omega}_{\nu,n}\|_V\|\eta^{\omega}_{\tilde{\nu},n}\|_V+2\frac{\nu^{\omega}(V)}{\nu^\omega Q^\omega(D)}\frac{\tilde{\nu}^{\omega}(V)}{\tilde{\nu}^\omega Q^\omega(D)}e^{-d  \lfloor n(\gamma^+-\beta)\rfloor/2}\prod_{i=0}^{n-1}Z(\theta^i\omega)^2
\end{multline}
where  $\rho_{C_d}:=\sup_{\omega\in K}\big\{1-\big(\epsilon_{C_d}^-(\omega)/\epsilon_{C_d}^+(\omega)\big)^2\big\}\in [0,1)$ and $Z(\omega):=\dfrac{1\vee\Upsilon(\omega)}{1\wedge\Psi(\omega)}$.
\end{prop}

We next describe how our assumptions compare to those of \citet[Theorem 1]{Douc2009}, who established a result in t.v.-norm of the form \eqref{eq:DM_result_front}, and how our Theorem \ref{thm:conv_to_zero} differs to a result of \citet{Whiteley2013} which places restrictive conditions on the observation sequence. Due to the technical nature of the assumptions and variations in notation, an exhaustive comparison would be very lengthy and tedious, so we just focus on some key issues.

\subsubsection*{Comparison with \citet{Douc2009}}

 Although \cite{Douc2009} addressed stability of the filtering distributions,  comparison of assumptions is most notationally direct in the setting \eqref{eq:pred_filters}; all assertions in the remainder of Section \ref{sub:statements} are to be understood in that context.

 The main feature of our assumptions which is stronger than those of \citet[Theorem 1]{Douc2009},  is that in \ref{H:drift}a) we require $C_d$ to be an LD set satisfying the integrability condition $\E\big[\log^-\big(\epsilon_{C_d}^-\mu_{C_d}(C_d\cap D)\big)\big]<+\infty$ \emph{for all} $d\in[\underline{d},+\infty)$. This is in contrast to \citet[Theorem 1, eq. (14)]{Douc2009}, which requires that a similar condition is satisfied for only \emph{some} LD set. The key place in which we use this integrability condition is in the proof of Proposition \ref{prop:Vnorm}, in particular see equation \eqref{eq:T_d_inf} below, where ultimately it helps us establish that for $\P$-almost all $\omega$, $\eta_{\nu,n}^\omega(V)$ cannot grow ``too-quickly'' as $n\to+\infty$.

 Otherwise, our assumptions are very similar to those of \citet[Theorem 1]{Douc2009}. We note that we have taken $Q^\omega(x,\setX)>0$ for all $x$, $\P$-a.s. by definition i.e. $g(x,Y_0(\omega))>0$ for all $x$, $\P$-a.s., which is essentially the same as \citet[p.139, condition (H1)]{Douc2009}.  Part b) of \ref{H:drift} is very similar to \citet[p.139, condition (H2)]{Douc2009}. We note in passing that similar conditions have appeared outside of the context of filtering, in the spectral theory of nonnegative kernels, see \cite{whiteley2012linear} and references therein.

  We note that \ref{H:theta} amounts to saying that the observation process $(Y_n)_{n\in\naturals}$ is stationary and ergodic. Combined with the conditions $\mathbb{E}[\log^+ \Upsilon]<\infty$ in \ref{H:Upsilon}, $\mathbb{E}[\log^- \Psi]<\infty$ in \ref{H:Psi},  and $\P(K)>2/3$ in \ref{H:setK}, this implies that $\limsup_{n\to\infty}n^{-1}\sum_{k=0}^{n-1}\log \Upsilon(\theta^k\omega) < +\infty$,  $\liminf_{n\to\infty}n^{-1}\sum_{k=0}^{n-1}\log \Psi(\theta^k\omega) >-\infty$ and $\lim_{n\to\infty}n^{-1}\sum_{k=0}^{n-1}\ind_K(\theta^k\omega)>2/3$, which are very similar to \citet[Theorem 1, conditions (12)-(14)]{Douc2009}. Motivation for the technical condition $\P(K)>2/3$ is given in \citet[Remark 5]{Douc2014}.

\subsubsection*{Comparison with \cite{Whiteley2013}}
A form of forgetting of the initial condition in $V$-norm for the prediction filters was established by the second author of the present work in \citet{Whiteley2013}, but under restrictive assumptions on the observation sequence. In the notation of the present work,   \citet[Corollary 1]{Whiteley2013} establishes under certain conditions that there exist $\overline{\setY}\subseteq\setY$ and constants  $c<+\infty$, $\rho<1$ depending on $\overline{\setY}$  such that
\begin{equation}
y\in\overline{\setY}^{\naturals} \quad \Rightarrow\quad \|\overline{\Pi}_{n}^\lambda(y,\cdot)-\overline{\Pi}_{n}^{\tilde\lambda}(y,\cdot)\|_V \leq C \rho^n,\quad \forall n\in\mathbb{N}.\label{eq:NW_result}
\end{equation}
\citet[Section 3.1.1.]{Whiteley2013} provides an example for which one can take $\overline{\setY}= \setY$ , but in other cases one must resort to strict inclusion $\overline{\setY}\subset\setY$ and the condition $y\in\overline{\setY}^{\naturals}$ becomes very restrictive. For instance in the setting of  \citet[Section 3.1.2.]{Whiteley2013}, $\setY=\mathbb{R}^{d_y}$, $\overline{\setY}$ is a compact set  and one can easily construct situations in which $\mathbb{P}(Y\in \overline{\setY}^{\naturals})=0$ when $\mathbb{P}$ is the law of $Y=(Y_n)_{n\in\naturals}$ under the correctly-specified HMM. Thus \eqref{eq:NW_result} does not satisfactorily extend \eqref{eq:DM_result_front}. Theorem \ref{thm:conv_to_zero} overcomes this deficiency.

\section{Discussion} \label{sub:examples}
The examples below serve two main purposes. Firstly, we show that for one of the models treated by \citet{Douc2009}, the tv-norm convergence as in \eqref{eq:DM_result_front} can be extended to convergence in $V$-norm with no further assumptions.  Secondly, we provide a simple example to illustrate  that under certain conditions on $g$, the filters can forget their initial condition in $V$-norm for some $V$ such that the $V$-norm of each prediction filter is infinite.

\subsection{A nonlinear state-space model}\label{sec:nonlinear_ssm}

Throughout Section \ref{sec:nonlinear_ssm} we take $\setX=\R^{d_x}$, $\setY=\R^{d_y}$ and we focus on the following nonlinear model, for $n\geq0$:
\begin{align}
X_{n+1}&=X_n+b(X_n)+\Sigma(X_n)V_n \label{eq:signal}\\
Y_n &= h(X_n)+\beta W_n\label{eq:obs}
\end{align}
where $b:\R^{d_x}\to \R^{d_x}$ and $h:\R^{d_x}\to\R^{d_y}$ are vectors of functions,  $\Sigma$ is a $d_x\times d_x$ matrix of continuous functions, $\beta>0$ is a constant and $(V_n)_{n\in\naturals}$ and $(W_n)_{n\in\naturals}$ are sequences of i.i.d. standard Gaussian vectors of appropriate dimension. The following conditions are considered by \citet[p.1245]{Douc2009}:
\begin{enumerate}[label=(\subscript{E}{\arabic*})]
\item\label{H:signal_B}$b$ is locally bounded and $\lim_{r\rightarrow+\infty}\sup_{|x|\geq r}\big|x+b(x)|-|x|=-\infty;$
\item\label{H:signal_sig} With $\sigma(\tau,x):=\tau^T\Sigma(x)\Sigma(x)^T\tau$,
\begin{align}
0<\inf_{(x,\tau)\in\mathbb{R}^{2d_x},\,|\tau|=1}\sigma(\tau,x)\leq \sup_{(x,\tau)\in\mathbb{R}^{2d_x},\,|\tau|=1}\sigma(\tau,x)<+\infty;
\end{align}
\item \label{H:obs} $h$ is locally bounded and $\limsup_{|x|\to+\infty} |x|^{-1}\log |h(x)|<+\infty$.
\end{enumerate}

\begin{remark}\label{rem:f}
 For an arbitrarily chosen $c>0$, set $V(x)=\exp(c|x|)$. Let $f$ be the Markov transition kernel corresponding to the signal model \eqref{eq:signal}. The following facts are gathered together from \citet[][p.1246]{Douc2009}. Under \ref{H:signal_B}-\ref{H:signal_sig}: there exists a constant $M<+\infty$ such that
 \begin{equation}
 \frac{f(V)(x)}{V(x)}\leq M\exp[c(|x+b(x)|-|x|)],\quad\forall x\in\setX,\label{eq:f_drift}
 \end{equation}
 and for any bounded Borel set  $C\in\sigX$ of strictly positive Lebesgue measure, there are constants $0<\tilde{\epsilon}^-_C\leq\tilde{\epsilon}^+_C<+\infty$ such that:
 \begin{equation}
 \tilde{\epsilon}^-_C\,\tilde{\mu}_{C}(A\cap C)\leq f(x,A\cap C )\leq  \tilde{\epsilon}^+_C\,\tilde{\mu}_{C}(A\cap C),\quad \forall (x,A)\in C\times\sigX,\label{eq:f_LD}
 \end{equation}
 where $\tilde{\mu}_C$ is the normalized restriction of Lebesgue measure to $C$. The Markov chain $(X_n)_{n\in\naturals}$ with transition kernel $f$ is aperiodic and positive Harris recurrent with unique invariant distribution, say $\pi$, such that $\pi(V)<+\infty$, and the bi-variate process $(X_n,Y_n)_{n\in\naturals}$ given by \eqref{eq:signal}-\eqref{eq:obs} is also aperiodic and positive Harris recurrent, with invariant distribution $\pi(\dd x)g(x,y)\dd y$ where $\dd y$ is Lebesgue measure on $\R^{d_y}$ and
 \begin{equation}
 g(x,y)\propto\exp(-[y-h(x)]^T[y-h(x)]/2\beta^2).\label{eq:g_example}
 \end{equation}
\end{remark}

The following proposition is an application of Theorem \ref{thm:conv_to_zero}.


\begin{prop}\label{prop:NLSS}
Assume \ref{H:signal_B}-\ref{H:obs} hold for the nonlinear state-space model. Let $\P$ be the probability measure on $\sigY^{\otimes\naturals}$ which is the law of $(Y_n)_{n\in\naturals}$ when the bi-variate process $(X_n,Y_n)_{n\in\naturals}$ satisfies  \eqref{eq:signal}-\eqref{eq:obs} and $X_0\sim\pi$. Then for any constant $c>0$ there exists  a  constant $\rho\in(0,1)$ such that, with $V(x)=\exp(c|x|)$,
$$
\lim_{n\to+\infty} \rho^{-n} \|\Pi_n^{\lambda}(Y,\cdot)-\Pi_n^{\tilde{\lambda}}(Y,\cdot)\|_{V} =0,\quad \mathbb{P}-a.s.
$$
for any two probability measures $\lambda,\tilde{\lambda}$ such that $\lambda(V)<+\infty$ and $\tilde{\lambda}(V)<+\infty$.
\end{prop}


\begin{proof}
 Let $V(x)=\exp(c|x|)$ with some arbitrary $c>0$. Fix any two probability measures $\lambda,\tilde{\lambda}$ on $\sigX$ such that $\lambda(V)\vee\tilde{\lambda}(V)<+\infty$. Consider the scenario \eqref{eq:filters}, let $\nu,\tilde{\nu}$ be the probability kernels associated with $\lambda,\tilde{\lambda}$ as per \eqref{eq:filters}, so $\eta_{\nu,n}^\omega(\cdot) \equiv \Pi_n^\lambda(Y(\omega),\cdot)$, $\eta_{\tilde{\nu},n}^\omega(\cdot) \equiv \Pi_n^{\tilde{\lambda}}(Y(\omega),\cdot)$.  To apply Theorem \ref{thm:conv_to_zero} we need to verify \ref{H:theta}-\ref{H:setK} and check that $\nu,\tilde{\nu}$ are members of $\mathcal{M}(D,V)$.

For \ref{H:theta}, the measure preservation part holds since $(Y_n)_{n\in\naturals}$ is by assumption a stationary process under $\P$. For the ergodicity part, by Remark \ref{rem:f}, the $(X_n,Y_n)_{n\in\naturals}$ chain is aperiodic and positive Harris recurrent, so by \citet[Theorem 2.6, Chapter 6, p.167]{revuz1974} the tail $\sigma$-algebra for the process $(X_n,Y_n)_{n\in\naturals}$ is a.s. trivial,  from which it follows that the  $\sigma$-algebra of events which are invariant w.r.t. the shift operator $\theta$, i.e., $\{A\in\F:\theta^{-1}(A)=A\}$,
is $\P$-trivial.

\ref{H:Upsilon} readily holds, since it follows from \ref{H:signal_B}, \eqref{eq:f_drift} and \eqref{eq:g_example}    that
$$\sup_{\omega} \Upsilon(\omega)=\sup_{\omega,x} Q^\omega(V)(x)/V(x)\leq\sup_{x,y}g(x,y)\sup_xf(V)(x)/V(x)<+\infty.$$

Consider now \ref{H:Psi} and \ref{H:drift}. For brevity, write
$$
\psi(x):=c(|x+b(x)|-|x|)+\log M+\log \sup_{x',y}g(x',y),\quad x\in\setX
$$
and then set $W(x)=0\vee-\psi(x)$. It follows from \ref{H:signal_B} that $\lim_{r\to\infty}\inf_{|x|\geq r}W(x)=+\infty$, therefore for any $d\in[0,+\infty)$, the set $C_d=\{x:W(x)\leq d\}$ is bounded. There must exist $\underline{d}\in[0,+\infty)$ such that $\{x:|x|\leq 1\}\subseteq C_{\underline{d}}$, otherwise $W$ would not be locally bounded, which would contradict the local boundedness of $b$ in \ref{H:signal_B}. Thus for each  $d\in[\underline{d},+\infty)$, $C_d$  is a bounded Borel set of strictly positive Lebesgue measure. Set $D=C_{\underline{d}}$.


Let $\overline{\setY}\in\sigY$ be any compact set and take $K=\{\omega:Y_1(\omega)\in \overline{\setY}\}$. For part a) of \ref{H:drift}, using \eqref{eq:f_LD}, we find that $C_d$ is a LD-set for $Q$ with:
$$
\epsilon_{C_d}^-(\omega)=\tilde{\epsilon}_{C_d}^-\inf_{x\in C_d}g(x,Y_1(\omega)), \quad  \epsilon_{C_d}^+(\omega)=\tilde{\epsilon}_{C_d}^+\sup_{x\in C_d}g(x,Y_1(\omega)),
$$
and $\mu_{C_d}^{\omega}(\cdot)=\tilde{\mu}_{C_d}(\cdot)$. Since $\tilde{\epsilon}^-_C\leq\tilde{\epsilon}^+_C$, we have $\epsilon_{C_d}^-(\omega)\leq \epsilon_{C_d}^+(\omega)$ for all $\omega\in\Omega$, as required. We also have $\inf_{\omega\in K}\big(\epsilon_{C_d}^-(\omega)/\epsilon_{C_d}^+(\omega)\big)\in (0,1]$ since $\overline{\setY}$ is compact, $C_d$ is bounded and $h$ is locally bounded under \ref{H:obs}.


To complete the verification of part a) of \ref{H:drift}, it remains to check that for any $d\geq\underline{d}$,
\begin{equation}
\mathbb{E}\Big[\log^-\big(\tilde{\epsilon}_{C_d}^- \tilde{\mu}_{C_d}(D)\inf_{x\in C_d }g(x,Y_1)\big)\Big]<+\infty,\label{eq:E_finite}
\end{equation}
where we note that $\tilde{\epsilon}_{C_d}^-$, $ \tilde{\mu}_{C_d}(D)$ are strictly positive constants by construction. Using the facts that: for any $a,b>0$, $\log^-(ab)\leq \log^-(a) + \log^-(b)$; $[y-h(x)]^T[y-h(x)]\leq 2 (|y|^2 +  |h(x)|^2)$; and, since $h$ is locally bounded, $\sup_{x\in C_d}|h(x)|^2<+\infty$; to establish \eqref{eq:E_finite} it suffices to show that $\mathbb{E}\left[|Y_1|^2\right]<+\infty$, or equivalently,
\begin{equation}
\int_{\setX}\int_{\setY}g(x,y)|y|^2 \,\dd y \,\pi(\dd x)<+\infty.\label{eq:int_finite}
\end{equation}
To establish \eqref{eq:int_finite} we follow \citet[][p.1246]{Douc2009}. Let $c^*>0$ and $V^*(x)=\exp(c^*|x|)$. Then, elementary manipulations give $\int_{\setY}g(x,y)|y|^2 \,\dd y = |h(x)|^2+\text{const.}$, and $\sup_x(|h(x)|^2/V^*(x))<+\infty$ by \ref{H:obs} as long as $c^*>2\limsup_{|x|\to+\infty}|x|^{-1}\log|h(x)|$, which we may assume since $c^*>0$ was arbitrary. By Remark \ref{rem:f}, $\pi(V^*)<+\infty$, so \eqref{eq:int_finite} holds, and then \eqref{eq:E_finite} does too, completing the verification of part a) of \ref{H:drift}. It is easily checked that the conditions of \ref{H:drift}b) hold by construction of $V$, $W$, and  $C_d$ and by \ref{H:signal_B}.

To verify \ref{H:Psi}, recall we have taken $D=C_{\underline{d}}$ and apply \eqref{eq:E_finite} with $d=\underline{d}$. \ref{H:setK} is easily achieved since  $\overline{\setY}\in\sigY$ was an arbitrary compact set.

Finally, we need to check $\nu,\tilde{\nu}\in \mathcal{M}(D,V)$. Since $\sup_{x,y}g(x,y)<+\infty$ and $g(x,y)>0$, $\int_{\setX} g(x,Y_0(\omega))\lambda(\dd x)\in(0,+\infty)$, hence $\nu^\omega(\cdot)$ is a probability measure on $\sigX$ for all $\omega$. The measurability of $\omega\mapsto \nu^\omega(A)$ is immediate. By assumption $\lambda(V)<+\infty$, hence $\nu^\omega(V)\leq\sup_{x,y}g(x,y)\lambda(V)/\int_{\setX} g(x,Y_0(\omega))\lambda(\dd x)<+\infty$ for all $\omega$, and $\nu^\omega Q^\omega (D)>0$ for all $\omega$ since $g(x,y)>0$ and $f(x,\cdot)$ has a strictly positive density w.r.t. Lebesgue measure. The same arguments apply to $\tilde{\nu}$.

\end{proof}

\subsection{Stability for filters but not for prediction filters}

It can be shown by arguments almost identical to those in the proof of Proposition \ref{prop:NLSS} that the claim of that proposition also holds with the prediction filters $\overline{\Pi}_n^\lambda,\overline{\Pi}_n^{\tilde{\lambda}}$ in place of the filters $\Pi_n^\lambda,\Pi_n^{\tilde{\lambda}}$. We omit the details to avoid repetition.  However, as we shall illustrate next, if $V(x)$ grows suitably quickly as $|x|\to\infty$, it can occur that the filters are stable in $V$-norm where as the prediction filters are not, in the sense of the following proposition. To demonstrate this phenomenon with a simple and short proof, we consider a specific linear state-space model. The result could easily generalized to a broader class of models, at the expense of a proof which involves lengthier technical manipulations.

\begin{prop}
Consider the model of Section \ref{sec:nonlinear_ssm} in the specific case $d_x=d_y=1$, and
\begin{align}
X_{n+1}&=\alpha X_n+ V_n\label{eq:linmodel_sig} \\
Y_n &= X_n+W_n\label{eq:linmodel_obs}
\end{align}
where $|\alpha|< 1$. Let $\P$ be the probability measure on $\sigY^{\otimes\naturals}$ which is the law of $(Y_n)_{n\in\naturals}$ when the bi-variate process $(X_n,Y_n)_{n\in\naturals}$ satisfies  \eqref{eq:signal}-\eqref{eq:obs} and $X_0\sim\pi$. Then there exist constants $c \in(1,2)$ and $\rho\in(0,1)$ such that, with  $V(x)=\exp(c\,x^2/2)$,
$$
\|\overline{\Pi}^\lambda_n(y,\cdot)\|_V = +\infty,\quad\forall (n,y)\in\mathbb{N}\times \setY^\naturals,
$$
for any probability measure $\lambda$, whereas
$$
\lim_{n\to+\infty} \rho^{-n} \|\Pi_n^{\lambda}(Y,\cdot)-\Pi_n^{\tilde{\lambda}}(Y,\cdot)\|_{V} =0,\quad \mathbb{P}-a.s.
$$
for any two probability measures $\lambda,\tilde{\lambda}$ such that $\lambda(V)<+\infty$ and $\tilde{\lambda}(V)<+\infty$.
\end{prop}
\begin{proof}
First note that for any $c\in(1,2)$ and $x\in\setX$,
$$
\int_{\setX} f(x,\dd z)V(z) \propto \int_{\setX}  \exp \left[ \frac{z^2}{2}(c-1) +\alpha zx - \frac{\alpha^2x^2}{2}  \right] \dd z =+\infty,
$$
where in the middle expression, $\dd z$ denotes Lebesgue measure on $\mathbb{R}$,
hence $\|\overline{\Pi}^\lambda_n(y,\cdot)\|_V = +\infty$ as claimed.

The proof is completed by applying Theorem \ref{thm:conv_to_zero} in the scenario \eqref{eq:filters}. In verifying \ref{H:theta}-\ref{H:setK} and checking that $\nu,\tilde{\nu}$ are members of $\mathcal{M}(D,V)$, where $\nu,\tilde{\nu}$ are the probability kernels associated with $\lambda,\tilde{\lambda}$ as per \eqref{eq:filters}, we can re-use some but not all of the arguments in the proof of Proposition \ref{prop:NLSS}.

Condition  \ref{H:theta} is verified exactly as in the proof of Proposition \ref{prop:NLSS}. For the condition \ref{H:Upsilon}, it follows by elementary manipulations that
$$
V(x)^{-1}\int f(x,\dd z)g(z,y)V(z) = \exp\left[\psi(x,y)\right],
$$
where
$$
\psi(x,y):=-\kappa\frac{\alpha^{2}x^{2}}{2}+\frac{\alpha xy}{2-c}+\frac{y^{2}}{2}\left(\frac{1}{2-c}-1\right)-\frac{\log 2\pi+\log(2-c)}{2},
$$
and we assume $c\in(1,2)$ is such that $0<1+c/\alpha^{2}-1/(2-c)=:\kappa$; note that it is easily checked that such a $c\in(1,2)$ indeed exists for any $\alpha\in(-1,1)$. Also,
$$
\Upsilon(\omega)\propto \exp\left[ \tilde{c}\, Y_{1}(\omega)^2 \right],
$$
where $\tilde{c}$ is a finite constant depending on $c$ and $\alpha$. It is easily seen from \eqref{eq:linmodel_sig} that $\pi$ is Gaussian, and hence from  \eqref{eq:linmodel_obs} that the law of $Y_1$ under $\P$ is also Gaussian, so $\mathbb{E}[\log^+ \Upsilon]<+\infty$, as required for  \ref{H:Upsilon}.

Take $\overline{\setY}\subset\setY$ as any compact set and set $W(x)=0 \vee -\sup_{y\in\overline{\setY}}\psi(x,y)$. Then $W$ is locally bounded and $\lim_{r\to\infty}\inf_{|x|\geq r}W(x)=+\infty$. The definitions and arguments used in verifying conditions \ref{H:Psi}-\ref{H:setK} then follow exactly as in the proof of Proposition \ref{prop:NLSS}, as does the verification that $\nu,\tilde{\nu}$ are members of $\mathcal{M}(D,V)$.
\end{proof}

\subsection{Concluding remarks}\label{sub:concluding_remarks}

We close with an outline of possible further applications and extensions of our results. For some HMM's, such as the linear-Gaussian state-space model with low-dimensional noise in \citet[Section 4.1]{Douc2014}, the Markov transition kernel $f$ may have a singular component, so that the LD part of \ref{H:drift}a) cannot be satisfied in the scenarios \eqref{eq:filters} or \eqref{eq:pred_filters}, but for some $m>1$ the $m$-fold iterate of $f$ admits a density w.r.t. a $\sigma$-finite measure. In such a situations, one may consider as an alternative to \eqref{eq:pred_filters}, the scenario in which again  $\Omega = \setY^\naturals$, but $\theta$ is the $m$-step shift $Y(\theta \omega)=(Y_m(\omega),Y_{m+1}(\omega),\ldots)$ and
$$
\nu^\omega=\lambda,\quad Q^\omega(x,A) = g(x,Y_0(\omega))\int_{\setX^m} \ind_A(x_m) f(x,dx_1)\prod_{k=2}^m f(x_{k-1},dx_k)g(x_{k-1},Y_{k-1}(\omega)),
$$
so then $\eta_{\nu,n}^\omega$ coincides with $\overline{\Pi}_{mn}^\lambda(Y(\omega),\cdot)$. Should the appropriate assumptions be satisfied, Theorem \ref{thm:conv_to_zero} would establish stability w.r.t. initial conditions for the subsequence $\{\overline{\Pi}_{n}^\lambda(Y,\cdot);n=km,k\in\naturals\}$.  Further investigation of this matter is left as a potential topic of future research.

Another possible avenue of investigation is to extend analysis to a two-sided time horizon,  $\Omega=\setY^\mathbb{Z}$, so then  $\theta$ the shift operator on $\Omega$ is invertible, and leverage our results to address the existence and uniqueness of a limiting probability kernel,  $\eta_\star^\omega(\cdot)= \lim_{n\to +\infty}\eta_{\nu,n}^{\theta^{-n}\omega}(\cdot)$, which in e.g. the scenario \eqref{eq:filters} may be regarded as a conditional distribution of $X_0$ given $(Y_{-n})_{n\in\naturals}$. Probability kernels of this form arise in the study of maximum likelihood estimators for HMM's, see \citep{Douc2012} and references therein. More abstractly, probability kernels of this form arise as generalizations of the Perron-Frobenius eigen-measure in the works of \cite{kifer1996perron} concerning large deviations for Markov chains in random environments, and \cite{whiteley2014twisted} concerning variance growth behaviour of sequential Monte Carlo approximations of marginal likelihoods for HMM's. Future research may address extension of the results of \cite{kifer1996perron} and \cite{whiteley2014twisted} under conditions similar to those in the present paper.

\section{Proofs and auxiliary results}\label{sec:proofs}
\subsection{Preliminaries}\label{sec:prelims}
We proceed with some further definitions. Define
\begin{equation}\label{eq:Mkernel}
G^{\omega}(x):=Q^{\omega}(x,\setX),\quad M^{\omega}(x,A)=\frac{Q^{\omega}(x,A)}{Q^{\omega}(x,\setX)},\quad (\omega,x,A)\in\Omega\times \setX\times \mathcal{X}.
\end{equation}


Throughout Sections \ref{sec:prelims} and \ref{sec:proof_of_vnorm}, we fix some $\nu\in\mathcal{M}(D,V)$ and define for $n\in\naturals$,
\begin{equation}
\lambda^{\omega}_n:=\eta_{\nu,n}^{\omega}(G^{\theta^{n}\omega}),\quad \omega\in\Omega\label{eq:lam_defn}
\end{equation}
and for $0\leq k\leq n$, functions  $h_{k,n}:\Omega\times\setX\rightarrow \mathbb{R}$ according to
\begin{equation}
h_{n,n}^{\omega}(x):= 1,\quad h^{\omega}_{k,n}(x):=\frac{Q_{n-k}^{\theta^k\omega}(x,\setX)}{\prod_{i=k}^{n-1}\lambda^{\omega}_i},\quad (\omega,x,k)\in \Omega\times\setX\times\{0,\dots,n-1\}.\label{eq:h_defn}
\end{equation}
Also for $1\leq k\leq n$, define $v_{k,n}:\Omega\times\setX\rightarrow \mathbb{R}$,
$$
v_{k,n}^{\omega}(x):=\frac{V(x)}{h_{k,n}^{\omega}(x)},\quad (\omega,x)\in \Omega\times\setX,
$$
with $V$ is as in \ref{H:drift}, and
\begin{equation}
S^{\omega}_{k,n}(x,A):=\frac{Q^{\theta^{k-1}\omega}(\ind_A h^{\omega}_{k,n})(x)}{\lambda^{\omega}_{k-1}h^{\omega}_{k-1,n}(x)},\quad (\omega,x,A)\in \Omega\times \setX\times\mathcal{X}.\label{eq:S_defn}
\end{equation}

\begin{lemma}\label{lemma:lambda}
Assume \ref{H:theta}, \ref{H:Upsilon}, \ref{H:Psi} and let $\nu\in\mathcal{M}(D,V)$. Then there exists $\bar{\Omega}\in\mathcal{F}$ with $\P(\bar{\Omega})=1$ such that the following hold for all $\omega\in \bar{\Omega}$. For all $0\leq k \leq n$,  $\lambda_n^\omega\in(0,+\infty)$, $\| h^\omega_{k,n} \|_V<+\infty$ and for all $x\in\setX$, $h^\omega_{k,n}(x)>0$.  For all $1\leq k\leq n$ and $x\in\setX$, $v^{\omega}_{k,n}(x)\in (0,+\infty)$, $S_{k,n}^{\omega}(x,\cdot)$ is a probability measure on $\sigX$, and
\begin{equation}
\eta_{\nu,n}^\omega(A) = \int_{\setX}  (S_{1,n}^{\omega}\cdots S_{n,n}^{\omega})(x,A)h^\omega_{0,n}(x) \nu^\omega(dx),\quad \forall A\in\sigX.\label{eq:S_eta_id}
\end{equation}
\end{lemma}

\begin{proof}
Let $\bar{\Omega}\in\mathcal{F}$ be the event of probability $1$ in Lemma \ref{lemma:eta_prelim}. Pick any $\omega\in \bar{\Omega}$. Then for any $n\in\naturals$, $G^{\theta^n\omega}(x)=Q^{\theta^n\omega}(x,\setX)>0$ for all $x\in\setX$, and  $\inf_{x\in D}G^{\theta^n\omega}(x)=\inf_{x\in D }Q^{\theta^n\omega}(x,\setX)\geq \Psi(\theta^n\omega)>0$, hence $\lambda_n^\omega > 0$. Also $\lambda_n^\omega = \eta_{\nu,n}^\omega Q^{\theta^n\omega}(\setX)\leq \| \eta_{\nu,n}^\omega \|_V  \opnm Q^{\theta^n\omega} \opnm_V<+\infty$.

Then, again using Lemma  \ref{lemma:eta_prelim}, $\| h^\omega_{k,n} \|_V = \opnm Q_{n-k}^{\theta^k \omega} \opnm_V / \prod_{i=k}^{n-1}\lambda_i ^\omega<+\infty$ and the inequality $h^{\omega}_{k,n}(x)>0$ holds since $Q^\omega(x,\setX)>0$ for all $x\in\setX$. By definition, $V(x)\in[1,+\infty)$, so $\| h^\omega_{k,n} \|_V <+\infty$ implies $h^{\omega}_{k,n}(x)<+\infty$ for all $x$, and we have already established $h^{\omega}_{k,n}(x)>0$, hence $v^{\omega}_{k,n}(x)\in(0,+\infty)$.  It follows from \eqref{eq:h_defn} that $Q^{\theta^{k-1}\omega}(h^\omega_{k,n})=\lambda_{k-1}^\omega h^\omega_{k-1,n}$, so $S^{\omega}_{k,n}(x,\cdot)$ is a probability measure.

It is easily checked using \eqref{eq:S_defn} that $(S_{1,n}^{\omega}\cdots S_{n,n}^{\omega})(x,A)=Q_n^\omega(x,A)/[h^\omega_{0,n}(x)\prod_{i=0}^{n-1}\lambda_i^\omega ]$ and from \eqref{eq:eta_defn} and \eqref{eq:lam_defn}   that  $\prod_{i=0}^{n-1}\lambda_i^\omega =\nu^\omega Q_n^\omega (\setX)$, from which \eqref{eq:S_eta_id} follows.
\end{proof}

\subsection{Proof of Proposition \ref{prop:Vnorm}}\label{sec:proof_of_vnorm}


\begin{lemma}\label{lem:Skernel}
Assume \ref{H:theta}-\ref{H:drift} and set:
$$
T_d(\omega):=1\wedge \epsilon^{-}_{C_d}(\omega) \mu_{C_d}^{\omega}( C_d \cap D),\quad \omega \in \Omega.
$$
Then there exists $\bar{\Omega}\in\mathcal{F}$ such that $\P(\bar{\Omega})=1$ and for all $(\omega,x,d)\in \bar{\Omega}\times\setX\times[ \underline{d},+\infty)$ and any $1\leq k\leq n$,
$$
S^{\omega}_{k,n}(v^{\omega}_{k,n})(x)\leq \rho^{\omega}_{k,n}v^{\omega}_{k-1,n}(x)+B^{\omega}_{k,n},
$$
where
\begin{align*}
&\rho^{\omega}_{k,n}:=\frac{1\vee \Upsilon(\theta^{k-1}\omega)}{\lambda^{\omega}_{k-1}}e^{-d\ind(\theta^{k-1}\omega\in K)}<+\infty,\\
&B^{\omega}_{k,n}:=V_d\frac{1\vee \Upsilon(\theta^{k-1}\omega)}{T_d(\theta^{k-1}\omega)}\prod_{i=k}^{n-1}\frac{\lambda^{\omega}_i}{1\wedge \Psi(\theta^i\omega)}<+\infty,
\end{align*}
with the convention that the product is unity when $k=n$, and with the dependence of $\rho^{\omega}_{k,n}$ and $B^{\omega}_{k,n}$ on $d$ suppressed in the notation.
\end{lemma}

\begin{proof}
Let $\bar{\Omega}$ be the intersection between: i) a set of $\P$-probability $1$ on which all the inequalities in the statement of Lemma \ref{lemma:eta_prelim} hold and ii) the set of $\P$-probability $1$ in the statement of   Lemma \ref{lemma:lambda}. Pick any $\omega\in\bar{\Omega}$. Let $d\geq \underline{d}$, $x\in\setX$, $1\leq k\leq n$ and note that
$$
S^{\omega}_{k,n}(v^{\omega}_{k,n})(x)=\frac{Q^{\theta^{k-1}\omega}(V)(x)}{\lambda^{\omega}_{k-1} h^{\omega}_{k-1,n}(x)}.
$$
If  $x\not\in C_d$, we have under \ref{H:drift},
\begin{align*}
S^{\omega}_{k,n}(v^{\omega}_{k,n})(x)&= \frac{V(x)}{\lambda^{\omega}_{k-1} h^{\omega}_{k-1,n}(x)}Q^{\theta^{k-1}\omega}(V)(x)/V(x)\\
&=v^{\omega}_{k-1,n}(x)\frac{Q^{\theta^{k-1}\omega}(V)(x)/V(x)}{\lambda^{\omega}_{k-1}}\\
&\leq v^{\omega}_{k-1,n}(x)\rho^{\omega}_{k,n},
\end{align*}
where $\rho^{\omega}_{k,n}$ is finite by Lemma's \ref{lemma:eta_prelim} and \ref{lemma:lambda}.

Now replace $\bar{\Omega}$ by its intersection with the set of $\omega'$ such that $\epsilon^{-}_{C_d}(\theta^k\omega')\mu^{\theta^k\omega'}_{C_d}(C_d\cap D)>0$ for all $k$, the latter being a set of probability $1$ by \ref{H:theta} and \ref{H:drift}. Then let $\omega$ be any point in this new $\bar{\Omega}$.

If $x\in C_d$,
\begin{align*}
\lambda^{\omega}_{k-1}h^{\omega}_{k-1,n}(x)=Q^{\theta^{k-1}\omega}(h^{\omega}_{k,n})(x)&\geq\frac{\epsilon^{-}_{C_d}(\theta^{k-1}\omega) \mu_{C_d}^{\theta^{k-1}\omega}(\ind_{C_d}Q_{n-k}^{\theta^{k}\omega}(\setX))}{\prod_{i=k}^{n-1}\lambda^{\omega}_i}\\
& \geq \epsilon^{-}_{C_d}(\theta^{k-1}\omega) \mu_{C_d}^{\theta^{k-1}\omega}(C_d\cap D)\prod_{i=k}^{n-1}\frac{\Psi(\theta^i\omega)}{\lambda^{\omega}_i}\\
&>0,
\end{align*}
with the convention here and in the remainder of the proof that the products are unity when $k=n$, and when $k<n$,  $\prod_{i=k}^{n-1}\Psi(\theta^i\omega)>0$ by Lemma \ref{lemma:eta_prelim} and $\prod_{i=k}^{n-1}\lambda^{\omega}_i<+\infty$ by Lemma \ref{lemma:lambda}.

Consequently, for  $x\in C_d$,
\begin{align*}
S^{\omega}_{k,n}(v^\omega_{k,n})(x)&\leq \frac{Q^{\theta^{k-1}\omega}(V)(x)}{T_d(\theta^{k-1}\omega)}\prod_{i=k}^{n-1}\frac{\lambda^{\omega}_i}{\Psi(\theta^i\omega)}\\
&\leq V(x)\frac{\Upsilon(\theta^{k-1}\omega)}{T_d(\theta^{k-1}\omega)}\prod_{i=k}^{n-1}\frac{\lambda^{\omega}_i}{\Psi(\theta^i\omega)}<+\infty.
\end{align*}
To conclude the proof, note that $V_d<+\infty$ under \ref{H:drift} and so for all $x\in\setX$,
$$
S^{\omega}_{k,n}(v^{\omega}_{k,n})(x)\leq \rho^{\omega}_{k,n}v^{\omega}_{k-1,n}(x)+B^{\omega}_{k,n}<+\infty.
$$
\end{proof}

\begin{lemma}\label{lem:Echeck}
Assume \ref{H:theta}-\ref{H:drift} and let
$$
Z(\omega):=\frac{1\vee \Upsilon(\omega)}{1\wedge \Psi(\omega)},\quad \omega\in\Omega.
$$
Then, for any $d\geq \underline{d}$,  $n\geq 1$, we have, $\P$-a.s.,
\begin{align*}
\eta_{\nu,n}^\omega(V)&\leq \frac{\nu^\omega(V)}{\nu^{\omega}Q^\omega(D)}e^{-d I_{0,n-1}^{\omega}} \prod_{i=0}^{n-1} Z(\theta^i\omega)
+V_d\, e^d\sum_{k=1}^n \frac{e^{-d I^\omega_{k-1,n-1} } }{T_d(\theta^{k-1}\omega)}\prod_{i=k-1}^{n-1}Z(\theta^i\omega)<+\infty.
\end{align*}
where  for $p\leq q$, $I_{p,q}^\omega:=\sum_{i=q}^p \ind_K (\theta^{i}\omega)$, $T_d$ is as in Lemma \ref{lem:Skernel}.

\end{lemma}

\begin{proof}
Noting that $v^\omega_{n,n}=V$ and using Lemma \ref{lem:Skernel} with $\bar{\Omega}$ as therein, elementary manipulations show that for any $(\omega,x)\in\bar{\Omega}\times\setX$,
\begin{align}
(S_{1,n}^{\omega}\cdots S_{n,n}^{\omega})(V)(x)\leq v_{0,n}(x)\prod_{k=1}^n\rho^{\omega}_{k,n}+\sum_{k=1}^n B^{\omega}_{k,n}\prod_{i=k+1}^n\rho^{\omega}_{i,n}<+\infty,\label{eq:S_V_bound}
\end{align}
with the convention that the right-most product equals 1 when $k=n$.

For $0\leq k<n$ and still with $\omega\in\bar{\Omega}$,
\begin{align*}
\prod_{i=k+1}^n\rho^{\omega}_{i,n}=e^{-d I_{k,n-1}^{\omega}}\prod_{i=k}^{n-1} \frac{1\vee \Upsilon(\theta^i\omega)}{\lambda^{\omega}_i}<+\infty,
\end{align*}
$$
\prod_{i=0}^{n-1}\lambda^{\omega}_i=\nu^{\omega}Q_n^{\omega}(\setX)\geq \nu^{\omega}Q^\omega(D)\prod_{i=0}^{n-1}1\wedge\Psi(\theta^i\omega)>0,
$$
and for $1\leq k <n$,
\begin{align*}
B^{\omega}_{k,n}\prod_{i=k+1}^n\rho^{\omega}_{i,n}&= V_d \, e^{-d  I^{\omega}_{k,n-1}}\frac{1\vee \Upsilon(\theta^{k-1}\omega)}{T_d(\theta^{k-1}\omega)}
\prod_{i=k}^{n-1} Z(\theta^i \omega)\\
&\leq V_d\, e^d \,\frac{e^{-d  I^{\omega}_{k-1,n-1}}}{T_d(\theta^{k-1}\omega)}
\prod_{i=k-1}^{n-1} Z(\theta^i \omega)<+\infty,
\end{align*}
plugging these into \eqref{eq:S_V_bound}, multiplying by $h^\omega_{0,n}$, integrating w.r.t. $\nu^\omega$, and noting \eqref{eq:S_eta_id} and the fact that $\nu^\omega(h^\omega_{0,n})=1$ completes the proof.

\end{proof}

\begin{lemma}\label{lem:lower_bounded}
Let $(Y_n)_{n \geq 0}$ be a sequence of nonnegative, equi-distributed random variables defined on a common probability space. If the expected value of $\log^- Y_0$ is finite, then for any $\beta\in(0,1)$, $\inf_{n\geq 0} \beta ^{-n} Y_n >0$, a.s.
\end{lemma}

\begin{proof}
See  \citet[][Lemma 7]{Douc2012}.
\end{proof}

\begin{proof}[Proof of Proposition \ref{prop:Vnorm}]
Let
$$
U^{\omega}_{d,n}:=
\sum_{k=1}^n \frac{e^{-d I^\omega_{k-1,n-1} } }{T_d(\theta^{k-1}\omega)}\prod_{i=k-1}^{n-1}Z(\theta^i\omega),\quad (d,\omega,n)\in[\underline{d},+\infty)\times \Omega\times\mathbb{N}^+
$$
with $T_d$ as in Lemma \ref{lem:Skernel} and  $Z$ as in Lemma \ref{lem:Echeck}. Note that Lemma \ref{lem:Echeck} implies that for any $d\geq\underline{d}$,
\begin{equation}
U_{d,n}^\omega<+\infty,\quad\forall n\in \mathbb{N}^+,\quad\ass \label{eq:U_finite}
\end{equation}

We shall show that there exists a $d^*\geq \underline{d}$ such that, for any $\beta\in (0,1)$ and $d>d^*$,
\begin{equation}
\lim_{n\rightarrow+\infty}\beta^n U^{\omega}_{d,n}=0,\quad\ass\label{eq:U_to_zero}
\end{equation}
and
\begin{equation}
\lim_{n\rightarrow+\infty}\beta^n e^{-d I_{0,n-1}^{\omega}} \prod_{i=0}^{n-1} Z(\theta^i\omega)=0,\quad\ass,\label{eq:first_term_to_zero}
\end{equation}
which, combined with Lemma \ref{lem:Echeck}, are enough to establish  $\lim_{n\rightarrow+\infty}\beta^n \|\eta_{\nu,n}^{\omega}\|_{V}=0$, $\P$-a.s.

Under  \ref{H:Upsilon} and \ref{H:Psi},
$$
0\leq \E\big[\log Z \big]=\E\big[|\log Z |\big]=\E\big[\log^+ \Upsilon \big]+\E\big[\log^- \Psi \big]<+\infty,
$$
so with $l:=\E\big[\log Z \big]$ and $\gamma:=\P(K)$, by \ref{H:theta} and by Birkhoff's ergodic theorem,
\begin{align}
\xi_n^{\omega}&:=n^{-1} \sum_{k=0}^{n-1} \log Z(\theta^k\omega)-l\;\to\; 0,\quad\ass\label{eq:xi_to_zero}\\
\tilde{\xi}_n^{\omega}&:=n^{-1}I^{\omega}_{0,n-1}-\gamma\;\to\; 0,\quad\ass\label{eq:xi_tilde_to_zero}
\end{align}
both as $n\to+\infty$, where we note that $\gamma>0$ by hypothesis of the proposition.

Now define $d^*:=l/\gamma \vee \underline{d} $ and set arbitrarily $d>d^*$; the main ideas of the proof are to show that under this condition and the assumptions of the proposition, with probability 1,
the terms $e^{-d I^\omega_{k-1,n-1}} \prod_{i=k-1}^{n-1} Z(\theta^i\omega)$ and  $1/T_d(\theta^{k-1}\omega)$ appearing in $U_{d,n}^\omega$ cannot grow ``too fast''  as $n-k\to+\infty$ and $k\to+\infty$, respectively.

Let $\beta\in(0,1)$ be as in the statement of the lemma, and pick $c\in (0,d\gamma - l)$ and $\tilde{\beta}\in(\beta,1)$ such that
\begin{equation}
\frac{\beta}{\tilde{\beta}} \exp (2c)<1. \label{eq:beta_tilde}
\end{equation}

Note that $\E\big[|\log T_d|\big]=\E\big[\log^-\big(\epsilon_{C_d}^-\mu_{C_d}(C_d\cap D)\big)\big]<+\infty$  under \ref{H:drift}, so by Lemma \ref{lem:lower_bounded} and under \ref{H:theta} we have
\begin{equation}
\underline{T}_{d,\tilde{\beta}}^{\omega}:=\inf_{n\in\mathbb{N}}\tilde{\beta}^{-n}T_d(\theta^n\omega)>0,\quad\ass\label{eq:T_d_inf}
\end{equation}
and by \eqref{eq:xi_to_zero}-\eqref{eq:xi_tilde_to_zero}, there exists $N^\omega_{c,d}\in\mathbb{N}$ such that
\begin{equation}
n\geq N^\omega_{c,d} \;\Rightarrow\; |\xi_n^\omega-d\tilde{\xi}_n^\omega|\leq c,\quad \ass\label{eq:N_cd}
\end{equation}

Now let $\omega$ be any point in a set of probability $1$ on which \eqref{eq:U_finite}, \eqref{eq:xi_to_zero}, \eqref{eq:xi_tilde_to_zero},  \eqref{eq:T_d_inf} and \eqref{eq:N_cd} all hold. Since we are interested in the limit as $n\to+\infty$, we assume for the rest of the proof that $n>N^\omega_{c,d}+1$.

Consider the decomposition
$$
U_{d,n}^\omega = U_{d,n,	1}^\omega + U_{d,n,2}^\omega,
$$
\begin{align}
 U_{d,n,1}^\omega&:=\sum_{k=1}^{N_{c,d}^\omega}  \frac{e^{-d I^\omega_{k-1,n-1} } }{T_d(\theta^{k-1}\omega)}\prod_{i=k-1}^{n-1}Z(\theta^i\omega),\label{eq:U_1}\\
 U_{d,n,2}^\omega&:=\sum_{k=N_{c,d}^\omega+1}^n \frac{e^{-d I^\omega_{k-1,n-1} } }{T_d(\theta^{k-1}\omega)}\prod_{i=k-1}^{n-1}Z(\theta^i\omega).\label{eq:U_2}
\end{align}

To prepare to bound \eqref{eq:U_1} and \eqref{eq:U_2}, note that
\begin{align}
e^{-d I^\omega_{k-1,n-1} } \prod_{i=k-1}^{n-1}Z(\theta^i\omega) &= e^{-(n-k+1)(d\gamma-l)}e^{n(\xi_n-d\tilde{\xi}_n) +(k-1)(d\tilde{\xi}_{k-1}-\xi_{k-1})}\label{eq:IZ_prod}\\
&\leq e^{-(n-k+1)(d\gamma-l-c)}e^{2nc},\label{eq:IZ_prod_bound}
\end{align}
where the equality holds for any $k\leq n $ and the inequality holds if additionally $k > N^\omega_{c,d} $.

To bound $U_{d,n,1}^\omega$,
\begin{align*}
U_{d,n,1}^\omega&= \left(e^{-d I^\omega_{N_{c,d}^\omega,n-1}}\prod_{i=N_{c,d}^\omega}^{n-1}Z(\theta^i\omega)\right)\sum_{k=1}^{N_{c,d}^\omega}   \frac{e^{-d I^\omega_{k-1,N_{c,d}^\omega-1} } }{T_d(\theta^{k-1}\omega)}\prod_{i=k-1}^{N_{c,d}^\omega-1} Z(\theta^i\omega)\\
&\leq e^{-(n-N_{c,d}^\omega)(d\gamma-l-c)}e^{2nc}\, U_{d,N_{c,d}^{\omega}}^{\omega}\\
&\leq e^{2nc} U_{d,N_{c,d}^{\omega}}^{\omega}
\end{align*}
where \eqref{eq:IZ_prod_bound} has been used and where $U_{d,N_{c,d}^{\omega}}^{\omega}$
does not depend on $n$ and is finite by \eqref{eq:U_finite}.

For $U_{d,n,2}^\omega$, applying  \eqref{eq:IZ_prod_bound} gives,
\begin{align*}
U_{d,n,2}^\omega &\leq \sum _{k = N_{c,d}^\omega+1}^n \frac{\tilde{\beta}^{-(k-1)}}{\underline{T}_{d,\tilde{\beta}}^{\omega}} e^{-(n-k+1)(d\gamma-l-c)}e^{2nc}\\
&\leq \frac{e^{2nc}\tilde{\beta}^{-n}}{\underline{T}_{d,\tilde{\beta}}^{\omega}} \frac{1}{1-e^{-(d\gamma - l - c)}}.
\end{align*}

Combining these upper bounds for $U_{d,n,1}$ and $U_{d,n,2}$, and recalling \eqref{eq:beta_tilde},
\begin{align*}
\beta^n U_{n,d}^\omega &\leq \beta^n e^{2nc}\left(U_{d,N_{c,d}^{\omega}}^{\omega} +\frac{\tilde{\beta}^{-n}}{\underline{T}_{d,\tilde{\beta}}^{\omega}}  \frac{1}{1-e^{-(d\gamma - l - c)}}  \right)\\
& \to 0,\quad \text{as} \quad n\to+\infty,
\end{align*}
which completes the proof of \eqref{eq:U_to_zero}.

In order to establish \eqref{eq:first_term_to_zero} and thus complete the proof of the proposition,   \eqref{eq:IZ_prod} applied with $k-1=0$ gives:
\begin{align*}
\beta^n e^{-d I^\omega_{0,n-1} } \prod_{i=0}^{n-1}Z(\theta^i\omega) &\leq \beta^n e^{-n(d\gamma-l-c)},\\
& \to 0,\quad \text{as}\quad n \to +\infty.
\end{align*}

\end{proof}

\subsection{Proof of Proposition \ref{prop:forget}}\label{sec:proof_of_forget}

\begin{proof}[Proof of Proposition \ref{prop:forget}]

We first introduce some additional notation. For $\bar{x}:=(x,x')\in\setX^2=:\bar{\setX}$, let $\bar{V}(\bar{x}):=V(x)V(x')$, for  functions $\psi_1,\,\psi_2:\setX\rightarrow\mathbb{R}$, let $\psi_1\otimes\psi_2(\bar{x}):=\psi_1(x)\psi_2(x')$, and for any two measures $\mu_1,\mu_2$ let $\mu_1\otimes\mu_2$ denote their direct product. Then let $\bar{Q}^{\omega}(\bar{x},\cdot ):=Q^{\omega}(x,\cdot)\otimes Q^{\omega}(x',\cdot)$.

Let $\omega$ be any point in a set of probability $1$ one which all the inequalities in the statement of Lemma \ref{lemma:eta_prelim} hold and $\nu,\tilde{\nu}$ satisfy the properties associated with their memberships of $\mathcal{M}(D,V)$.  We keep this $\omega$ fixed throughout the proof,  so to slightly economise on notation we  suppress the dependence of  $\nu$ and $\tilde{\nu}$ on $\omega$.

Independently of $\omega$, fix $d\geq \underline{d}$ and $\varphi:\setX\rightarrow\mathbb{R}$ a measurable function such that $|\varphi|\leq V$. Then, with $\varphi^+\geq0$ and $\varphi^-\geq0$ being respectively the positive and negative parts of $\varphi$, i.e. $\varphi=\varphi^+ -\varphi^-$,
\begin{align*}
|\eta^{\omega}_{\nu,n}(\varphi)-\eta^{\omega}_{\tilde{\nu},n}(\varphi)|&\leq |\eta^{\omega}_{\nu,n}(\varphi^+)-\eta^{\omega}_{\tilde{\nu},n}(\varphi^+)|+ |\eta^{\omega}_{\nu,n}(\varphi^-)-\eta^{\omega}_{\tilde{\nu},n}(\varphi^-)|\\
&\leq \frac{|\Delta_{\nu,\tilde{\nu},n}^{\omega}(\varphi^+,\ind_{\setX})| + |\Delta_{\nu,\tilde{\nu},n}^{\omega}(\varphi^-,\ind_{\setX})|}{\nu\otimes\tilde{\nu} \bar{Q}_{n}^{\omega}(\bar{\setX})}
\end{align*}
where, for measurable functions $\psi_1,\,\psi_2:\setX\rightarrow\mathbb{R}^+$,
$$
\Delta^{\omega}_{\nu,\tilde{\nu},n}(\psi_1,\psi_2):=\nu Q_{n}^{\omega}(\psi_1)\tilde{\nu} Q_{n}^{\omega}(\psi_2)-\nu Q_{n}^{\omega}(\psi_2)\tilde{\nu} Q_{n}^{\omega}(\psi_1).
$$

Let  $\psi:\setX\rightarrow\mathbb{R}^+$ be any measurable function such that $\psi\leq V$. Then, following very similar arguments to \citet[][Proof of Proposition 5, pp.2712-2713]{Douc2012}, one obtains
\begin{align*}
\Big| \Delta&^{\omega}_{\nu,\tilde{\nu},n}(\psi,\ind_{\setX})\Big|\\
&\leq \int_{\bar{\setX}^{n+1}}  |\psi\otimes \ind_{\setX}(\bar{x}_n)-\ind_{\setX}\otimes\psi(\bar{x}_n)|\rho_{C_d}^{\sum_{j=0}^{n-1}\ind_{\bar{C}_d\times\bar{C}_d}(\bar{x}_{j},\bar{x}_{j+1})\ind_K(\theta^j\omega)} \\
&\quad\times\nu\otimes\tilde{\nu}(\dd\bar{x}_0) \prod_{i=0}^{n-1}\bar{Q}^{\theta^i\omega}(\bar{x}_{i},\dd\bar{x}_{i+1}).
\end{align*}
Now since, $V\geq 1$ and $0\leq \psi\leq V$,
$$
|\psi\otimes \ind_{\setX}(\bar{x})-\ind_{\setX}\otimes\psi(\bar{x})|\leq \psi\otimes \ind_{\setX}(\bar{x}) \vee \ind_{\setX}\otimes\psi(\bar{x}) \leq\bar{V}(\bar{x}),\quad\forall\bar{x}\in\bar{\setX}.
$$


Both $0\leq \varphi^-\leq V$ and $0 \leq \varphi^+\leq V$, so we may apply the above bounds to obtain:
\begin{multline}\label{eq:delta}
|\eta^{\omega}_{\nu,n}(\varphi)-\eta^{\omega}_{\tilde{\nu},n}(\varphi)|\nu\otimes\tilde{\nu} \bar{Q}_{n}^{\omega}(\bar{\setX})\\
\leq2 \int_{\bar{\setX}^{n+1}}\bar{V}(\bar{x}_n)\rho_{C_d}^{\sum_{j=0}^{n-1}\ind_{\bar{C}_d\times\bar{C}_d}(\bar{x}_{j},\bar{x}_{j+1})\ind_K(\theta^j\omega)}\nu\otimes\tilde{\nu}(\dd\bar{x}_0) \prod_{i=0}^{n-1}\bar{Q}^{\theta^i\omega}(\bar{x}_{i},\dd\bar{x}_{i+1}).
\end{multline}

Then, for any set $\bar{A}\in\sigX^{\otimes2}$ writing $M_{\bar{A},n}(\bar{x}_{0:n-1}):=\sum_{i=0}^{n-1}\ind_{\bar{A}}(\bar{x}_i)$ and following very similar arguments to those of of \citet[][Proof of Proposition 5, p.2714]{Douc2012}, we have under the hypothesis of the proposition on $I_{0,n-1}^\omega$, that for any $\beta\in(\gamma^-,1]$,
\begin{equation}
\rho_{C_d}^{\sum_{j=0}^{n-1}\ind_{\bar{C}_d\times\bar{C}_d}(\bar{x}_{j},\bar{x}_{j+1})\ind_K(\theta^j\omega)} \leq \rho_{C_d}^{\lfloor n (\beta-\gamma^-)\rfloor} + \ind\left\{M_{\bar{C}_d^c,n}(\bar{x}_{0:n-1})\geq (n-\lfloor n\beta\rfloor)/2\right\}.\label{eq:rho_bound}
\end{equation}
Substituting \eqref{eq:rho_bound} into \eqref{eq:delta} and noticing $\nu\otimes\tilde{\nu}\bar{Q}_n^\omega(\bar{\setX})\geq \nu Q^\omega(D)\tilde{\nu} Q^\omega(D)\prod_{i=0}^{n-1}1\wedge\Psi(\theta^k\omega)^2 $,
\begin{multline}
|\eta^{\omega}_{\nu,n}(\varphi)-\eta^{\omega}_{\tilde{\nu},n}(\varphi)| \\
  \leq 2 \rho_{C_d}^{\lfloor n (\beta-\gamma^-)\rfloor}\|\eta^{\omega}_{\nu,n}\|_V\|\eta^{\omega}_{\tilde{\nu},n}\|_V + \frac{2\Gamma_{\nu,\tilde{\nu},n}^\omega}{\nu Q^\omega(D)\tilde{\nu}Q^\omega(D)\prod_{i=0}^{n-1}1\wedge\Psi(\theta^k\omega)^2},\label{eq:eta_bound_gam}
\end{multline}
where
$$
\Gamma_{\nu,\tilde{\nu},n}^\omega:= \int_{\bar{\setX}^{n+1}}\bar{V}(\bar{x}_n)\ind\left\{M_{\bar{C}_d^c,n}(\bar{x}_{0:n-1})\geq (n-\lfloor n\beta\rfloor)/2\right\}\,\nu\otimes\tilde{\nu}(\dd\bar{x}_0) \prod_{i=0}^{n-1}\bar{Q}^{\theta^i\omega}(\bar{x}_{i},\dd\bar{x}_{i+1}).
$$
Re-writing $\Gamma_{\nu,\tilde{\nu},n}^\omega$,
\begin{align*}
\Gamma_{\nu,\tilde{\nu},n}^\omega & = \left(\prod_{i=0}^{n-1}1\vee\Upsilon(\theta^i\omega)^2\right)\int_{\bar{X}^{n+1}}\bar{V}(\bar{x}_0)\nu\otimes\tilde{\nu}(\dd\bar{x}_0) \ind\left\{M_{\bar{C}_d^c,n}(\bar{x}_{0:n-1})\geq (n-\lfloor n\beta\rfloor)/2\right\}\\
&\quad \times \left(e^{-d \sum_{i=0}^{n-1}\ind_{\bar{C}_d^c}(\bar{x}_i)\ind_K(\theta^i\omega)}\right)\prod_{i=0}^{n-1}\frac{\bar{Q}^{\theta^i\omega}(\bar{x}_i,\dd\bar{x}_{i+1})\bar{V}(\bar{x}_{i+1})}{\bar{V}(\bar{x}_{i})e^{-d\ind_{\bar{C}_d^c}(\bar{x}_i)\ind_K(\theta^i\omega)}1\vee\Upsilon(\theta^i\omega)^2}.
\end{align*}
Following very similar arguments to those of \citet[][Proof of Proposition 5, p.2715]{Douc2012}, under the hypothesis of the proposition on $I_{0,n-1}^\omega$ we have for any $\beta\in(0,\gamma^+)$,
$$
\left(e^{-d \sum_{i=0}^{n-1}\ind_{\bar{C}_d^c}(\bar{x}_i)\ind_K(\theta^i\omega)}\right)\ind\left\{M_{\bar{C}_d^c,n}(\bar{x}_{0:n-1})\geq (n-\lfloor n\beta\rfloor)/2\right\} \leq e^{-dn\lfloor(\gamma^+-\beta) \rfloor/2},
$$
and it follows from \ref{H:drift} that
$$
\sup_{\bar{x}_i\in\bar{\setX}}\frac{\int_{\bar{\setX}} \bar{Q}^{\theta^i\omega}(\bar{x}_i,\dd\bar{x}_{i+1})\bar{V}(\bar{x}_{i+1})}{\bar{V}(\bar{x}_{i})e^{-d\ind_{\bar{C}_d^c}(\bar{x}_i)\ind_K(\theta^i\omega)}1\vee\Upsilon(\theta^i\omega)^2} \leq 1.
$$
The proof is completed upon applying these last two inequalities to bound $\Gamma_{\nu,\tilde{\nu},n}^\omega$ in \eqref{eq:eta_bound_gam}.

\end{proof}

\subsection{Proof of Theorem \ref{thm:conv_to_zero}}

\begin{proof}[Proof of Theorem \ref{thm:conv_to_zero}]
By \ref{H:setK}, $\P(K)>2/3$, implying that there exist $0<\gamma^-<\gamma^+<1$ such that $\P(K)>(1-\gamma^-)\vee (1+\gamma^+)/2$ \citep[see][Remark 5]{Douc2014}. Consequently, under \ref{H:theta} and by Birkhoff's ergodic theorem, there exists $\underline{N}^{\omega}\in\mathbb{N}$ such that
$$
n\geq \underline{N}^{\omega} \;\Rightarrow\; n^{-1}|I^{\omega}_{0,n-1}|\geq  (1- \gamma^-)\vee (1+\gamma^+)/2,\quad\ass
$$
With $Z(\omega)$ as in Lemma \ref{lem:Echeck}, and under \ref{H:theta}, \ref{H:Upsilon}, \ref{H:Psi}, there exists $l \geq 0$ such that
\begin{equation}
\xi_n^\omega:= \frac{1}{n}\sum_{k=0}^{n-1}\log Z(\theta^k \omega)- l\;\to\; 0,\quad\text{as}\quad n \to+\infty, \quad\ass \label{eq:xi_n_thm1}
\end{equation}

Now fix any $\beta\in(\gamma^-,\gamma^+)$ and $d \geq \underline{d}$ such that
\begin{equation}
d(\gamma^+-\beta)/2 > 2l,\label{eq:d_thm1}
\end{equation}
and with $\rho_{C_d}:=\sup_{\omega\in K}\big\{1-\big(\epsilon_{C_d}^-(\omega)/\epsilon_{C_d}^+(\omega)\big)^2\big\}\in [0,1)$ as in Proposition \ref{prop:forget}, then also fix $\rho\in(0,1)$ such that
\begin{equation}
\rho > \rho_{C_d}^{\beta-\gamma^-} \vee e^{-d(\gamma^+ -\beta)/2+2l}. \label{eq:rho_thm1}
\end{equation}

By Proposition \ref{prop:forget}, for $\P$-almost any $\omega$ and $n\geq \underline{N}^{\omega}$,
\begin{align}
\rho^{-n}\|\eta^{\omega}_{\nu,n}-\eta^{\omega}_{\tilde{\nu},n}\|_V& \nonumber \\
&\leq  2\rho^{-n}\rho_{C_d}^{\lfloor n (\beta-\gamma^-)\rfloor}\|\eta^{\omega}_{\nu,n}\|_V\|\eta^{\omega}_{\tilde{\nu},n}\|_V \label{eq:rho_eta_diff1}\\
&+2\frac{\nu^{\omega}(V)}{\nu^{\omega}Q^\omega(D) } \frac{\tilde{\nu}^{\omega}(V)}{\tilde{\nu}^{\omega}Q^\omega(D) }\rho^{-n}e^{-d  \lfloor n(\gamma^+-\beta)\rfloor/2}\prod_{i=0}^{n-1}Z(\theta^i\omega)^2,\label{eq:rho_eta_diff2}
\end{align}
where $\nu^{\omega}(V)/\nu^\omega Q^\omega(D)<+\infty$ and $\tilde{\nu}^{\omega}Q^\omega(V)/\tilde{\nu}^\omega(D)<+\infty$, $\P$-a.s., since $\nu,\tilde{\nu}\in\mathcal{M}(D,V)$.

For the term in \eqref{eq:rho_eta_diff1},
\begin{align*}
&\rho^{-n} \rho_{C_d}^{\lfloor n (\beta-\gamma^-)\rfloor}\|\eta^{\omega}_{\nu,n}\|_V\|\eta^{\omega}_{\tilde{\nu},n}\|_V \\
& \leq  \rho_{C_d}^{-1}  \left(\frac{\rho_{C_d}^{\beta-\gamma^-} }{ \rho}\right)^{n/2}\|\eta^{\omega}_{\nu,n}\|_V \left(\frac{\rho_{C_d}^{\beta-\gamma^-} }{ \rho}\right)^{n/2}\|\eta^{\omega}_{\tilde{\nu},n}\|_V\\
&\to 0, \quad\text{as}\quad n\to+\infty,\quad \ass,
\end{align*}
where the convergence is due to \eqref{eq:rho_thm1} and Proposition \ref{prop:Vnorm}.

For the term in \eqref{eq:rho_eta_diff2},
\begin{align*}
\rho^{-n}e^{-d  \lfloor n(\gamma^+-\beta)\rfloor/2}\prod_{i=0}^{n-1}Z(\theta^i\omega)^2 &\leq e^{d/2}  \rho^{-n}  e^{-d   n(\gamma^+ -\beta)/2} \prod_{i=0}^{n-1}Z(\theta^i\omega)^2 \\
& = e^{d/2}  \rho^{-n} e^{- n( d (\gamma^+ -\beta)/2 -2 l) } e^{2n\xi_n}\\
&\to 0, \quad\text{as}\quad n\to+\infty,\quad \ass,
\end{align*}
where  the convergence is due to \eqref{eq:xi_n_thm1}, \eqref{eq:d_thm1} and \eqref{eq:rho_thm1}. The proof is complete.

\end{proof}

\section*{Acknowledgement}

Mathieu Gerber is supported by DARPA under Grant No. FA8750-14-2-0117.

\bibliographystyle{apalike}
\bibliography{complete}

\begin{thebibliography}{}

\bibitem[Billingsley, 1986]{Bill2ndEd}
Billingsley, P. (1986).
\newblock {\em Probability and measure}.
\newblock Wiley, 2nd edition.

\bibitem[Crauel, 2003]{Crauel2003}
Crauel, H. (2003).
\newblock {\em Random probability measures on {P}olish spaces}, volume~11.
\newblock CRC press.

\bibitem[Crisan and Rozovskii, 2011]{crisan2011oxford}
Crisan, D. and Rozovskii, B., editors (2011).
\newblock {\em The Oxford handbook of nonlinear filtering}.
\newblock Oxford University Press.

\bibitem[Douc et~al., 2009]{Douc2009}
Douc, R., Fort, G., Moulines, E., and Priouret, P. (2009).
\newblock Forgetting the initial distribution for hidden {M}arkov models.
\newblock {\em Stochastic processes and their applications}, 119(4):1235--1256.

\bibitem[Douc et~al., 2010]{douc2010forgetting}
Douc, R., Gassiat, E., Landelle, B., and Moulines, E. (2010).
\newblock Forgetting of the initial distribution for nonergodic hidden markov
  chains.
\newblock {\em Ann. Appl. Prob.}, 20(5):1638--1662.

\bibitem[Douc and Moulines, 2012]{Douc2012}
Douc, R. and Moulines, E. (2012).
\newblock Asymptotic properties of the maximum likelihood estimation in
  misspecified hidden {M}arkov models.
\newblock {\em Ann. Statist.}, 40(5):2697--2732.

\bibitem[Douc et~al., 2014]{Douc2014}
Douc, R., Moulines, E., and Olsson, J. (2014).
\newblock Long-term stability of sequential {M}onte {C}arlo methods under
  verifiable conditions.
\newblock {\em Ann. Appl. Prob.}, 24(5):1767--1802.

\bibitem[Kifer, 1996]{kifer1996perron}
Kifer, Y. (1996).
\newblock {Perron-Frobenius} theorem, large deviations, and random
  perturbations in random environments.
\newblock {\em Mathematische Zeitschrift}, 222(4):677--698.

\bibitem[Kleptsyna and Veretennikov, 2008]{kleptsyna2008discrete}
Kleptsyna, M. and Veretennikov, A. (2008).
\newblock On discrete time ergodic filters with wrong initial data.
\newblock {\em Probability Theory and Related Fields}, 141(3-4):411--444.

\bibitem[Meyn and Tweedie, 2009]{meyntweedie}
Meyn, S. and Tweedie, R. (2009).
\newblock {\em {Markov} chains and stochastic stability}.
\newblock Cambridge University Press, 2nd edition.

\bibitem[Ocone and Pardoux, 1996]{ocone1996asymptotic}
Ocone, D. and Pardoux, E. (1996).
\newblock Asymptotic stability of the optimal filter with respect to its
  initial condition.
\newblock {\em SIAM Journal on Control and Optimization}, 34(1):226--243.

\bibitem[Revuz, 1975]{revuz1974}
Revuz, D. (1975).
\newblock {\em Markov Chains}, volume~11 of {\em North-Holland Mathematical
  Library}.
\newblock North-Holland, 1st edition.

\bibitem[{Van Handel}, 2009]{van2009stability}
{Van Handel}, R. (2009).
\newblock The stability of conditional {Markov} processes and {Markov} chains
  in random environments.
\newblock {\em Ann. Appl. Prob.}, 37(5):1876--1925.

\bibitem[Whiteley, 2013]{Whiteley2013}
Whiteley, N. (2013).
\newblock Stability properties of some particle filters.
\newblock {\em Ann. Appl. Prob.}, 23(6):2500--2537.

\bibitem[Whiteley et~al., 2012]{whiteley2012linear}
Whiteley, N., Kantas, N., and Jasra, A. (2012).
\newblock Linear variance bounds for particle approximations of
  time-homogeneous {Feynman--Kac} formulae.
\newblock {\em Stochastic Processes and their Applications}, 122(4):1840--1865.

\bibitem[Whiteley and Lee, 2014]{whiteley2014twisted}
Whiteley, N. and Lee, A. (2014).
\newblock Twisted particle filters.
\newblock {\em Ann. Statist.}, 42(1):115--141.

\end{thebibliography}

\end{document}